\documentclass[11pt]{article}

\usepackage[linesnumbered,vlined]{algorithm2e}
\usepackage{amsmath}
\usepackage{amsfonts, amssymb}
\usepackage{cite}

\usepackage{amsthm}
\usepackage{fullpage}
\newtheorem{theorem}{Theorem} 
\newtheorem{lemma}[theorem]{Lemma}

\newenvironment{lemma-repeat}[1]{\begin{trivlist}
\item[\hspace{\labelsep}{\bf\noindent Lemma \ref{#1} }]\em }%
{\end{trivlist}}
\newenvironment{theorem-repeat}[1]{\begin{trivlist}
\item[\hspace{\labelsep}{\bf\noindent Theorem \ref{#1} }]\em }%
{\end{trivlist}}

\newcommand{\remove}[1]{}

\newcommand{\size}[1]{\ensuremath{\left|#1\right|}}
\newcommand{\set}[1]{\left\{ #1 \right\}}

\DeclareMathOperator{\polylog}{polylog}

\newcommand{\bbR}{\mathop{\mathbb{R}}}
\newcommand{\wbar}{\bar{w}}
\newcommand{\xbar}{\bar{x}}

\newcommand{\lrmwm}{Algorithm~\texttt{MWM-seq}}
\newcommand{\mwmsemi}{Algorithm~\texttt{MWM-semi}}

\begin{document}

\title{A $(2+\epsilon)$-Approximation for Maximum Weight Matching\\
        in the Semi-Streaming Model%
\thanks{An extended abstract of this work was presented in SODA 2017~\cite{PazS17}.}}
\author{Ami Paz\footnotemark[3]
\and Gregory Schwartzman\footnotemark[4]}

\date{} 

\maketitle

\renewcommand*{\thefootnote}{\fnsymbol{footnote}}
\footnotetext[3]{%
IRIF, CNRS and University Paris Diderot, France
\texttt{amipaz@irif.fr}.
Supported by the Fondation Sciences Math\'ematiques de Paris (FSMP)
}
\footnotetext[4]{%
National Institute of Informatics, Japan 
\texttt{greg@nii.ac.jp}. 
Supported by JSPS KAKENHI Grant Number JP18H05291.
}
\renewcommand*{\thefootnote}{\roman{footnote}}

\begin{abstract}
We present a simple deterministic single-pass
$(2+\epsilon)$-approximation algorithm for the maximum weight matching problem
in the semi-streaming model.
This improves upon the currently best known approximation ratio of $(4+\epsilon)$.

Our algorithm uses $O(n\log^2 n)$ bits of space for constant values of $\epsilon$.
It relies on a variation of the local-ratio theorem,
which may be of use for other algorithms in the semi-streaming model as well.
\end{abstract}

\section{Introduction}
We present a simple $(2+\epsilon)$-approximation algorithm
for the maximum weight matching (MWM) problem
in the semi-streaming model.
Our algorithm is deterministic, single-pass,
requires only $O(1)$ processing time per incoming edge,
and
uses $O(n \log^2 n)$ bits of space for any constant $\epsilon>0$.
This improves upon the previously best known
approximation algorithm 
of Crouch and Stubbs~\cite{CrouchS14},
which achieves an approximation ratio of $(4+\epsilon)$
and takes $O(\log n)$ time to process an edge.
Our main result is as follows.
\begin{theorem}
	\label{thm: the main theorem}
	There exists an algorithm in the semi-streaming model computing a $(2+\epsilon)$-approximation for MWM, using $O(\epsilon^{-1}n\log n \cdot (\log n + \log (1/\epsilon)))$ bits and having an $O(1)$ processing time.
\end{theorem}

The MWM problem is a classical problem in graph theory.
Its first efficient solution is due to Edmonds~\cite{Edmonds1965},
which was later improved by
Micali and Vazirani~\cite{MicaliV80}.
The MWM problem was one of the first to be considered
in the semi-streaming model
when this model was first presented~\cite{FeigenbaumKMSZ05},
and apparently the most studied problem in this model since
(see ``Related Work'').

In the first algorithms for the MWM problem in the semi-streaming model,
a matching is maintained at all times,
and is being updated according to the incoming edges.
More recent algorithms sort the edges
into weight classes,
keep a subset of each class,
and then find a matching in the union of these subsets.

Like previous algorithms,
our algorithm maintains a set of edges
from which the final matching is constructed;
however,
unlike some of the previous algorithms, 
we do not maintain a matching at all times,
but only construct it in a post-processing stage.
Our main technical contribution is the
adaptation of the local-ratio technique for
maximization problems~\cite{BarYehudaE1985, Bar-NoyBFNS01}
to the semi-streaming model,
in a novel and simple manner.
Our work presents
a significantly better approximation ratio
for the MWM problem,
along with a new approximation technique
for optimization problems in the semi-streaming model.

For the maximum \emph{unweighted} matching problem,
a simple greedy algorithm yields a $2$-approximation.
This was observed in the very first paper
on the semi-streaming model~\cite{FeigenbaumKMSZ05},
and not improved since.
Any future improvement of the approximation factor
to a constant smaller than $2$
will also solve this long-standing problem.

\paragraph{Our Contribution}
When developing an algorithmic framework for a new model, it is natural to first address the most fundamental algorithmic problems.
Finding a large matching in a graph is indeed a fundamental problem,
which has been extensively studied in the model of semi-streaming graph algorithms.
Our algorithm uses an extension of a well studied
approximation framework,
the local-ratio technique,
while previous algorithm used clever ideas which were
specifically crafted for the problem and model.

As noted, a simple greedy algorithm gives a 2-approximation
for MWM in the unweighted case.
In the weighted case,
a 2-approximation can be achieved by first sorting the edges
from the heaviest to the lightest,
and then adding them greedily to form a matching.
However, in the semi-streaming model it is impossible to keep a list of all the edges in the memory in order to sort them.
Instead, the local-ratio technique allows us to
ignore some of the edges,
and run a greedy algorithm on the remaining edges, in  an arbitrary order.
In this work,
we extend the local-ratio technique,
in a way that allows us to discard all but
$O(n\log n)$ of the edges,
complying with the memory restrictions of the model.

A simple local-ratio algorithm for the
MWM problem in the sequential model of computation
goes roughly as follows:
repeatedly select an edge with positive weight;
reduce its weight from the edge itself and
from all its neighboring edges;
push the edge into a stack
and continue to the next edge,
as long as there is an edge with positive weight;
finally,
unwind the stack and add the edges greedily
to the matching.
This procedure results in a 2-approximation for the MWM problem.
It can be extended to a $(2\alpha)$-approximation,
for $\alpha > 1$,
if at each step we reduce the weight of the processed edge
multiplied by $\alpha$ from its adjacent edges.

The challenge in translating this technique to the semi-streaming model is twofold.
First, we have to reduce edge weights from edges that are yet to arrive.
This is solved by saving,
for each node,
the total amount that should be reduced from each arriving
edge containing this node,
and reducing weight \emph{retroactively} from incoming stream edges.

The second, more substantial challenge,
is limiting the size of the stack,
so it can comply with the $O(n\polylog n)$
space bound.
It is not hard to come up with an execution of the above algorithm
where all edges are eventually stored
in the stack,
which may take $\Omega(n^2 \polylog n)$ bits of space.
To overcome this problem,
we \emph{remove edges} from within the stack,
\emph{during} the execution of the algorithm.
The traditional local-ratio technique was not designed to work under space limitations,
and thus does not guarantee any approximation ratio
if edges are removed from the stack.
The crux of our approach
is a variation of the local-ratio technique,
which provides conditions under which an edge may be removed from the stack
while incurring only a small loss
in the approximation ratio.

Specifically, we show that if an edge in the stack
is significantly lighter than its neighboring edge,
and this neighboring edge is added to the stack,
then removing the light edge has only a small effect on the total weight of the solution.
In order to use this conclusion,
we must first assure a steady increase
in the edge weights around each node.
This, in turn, requires another adaptation to the 
classical local-ratio approach for the problem.

To assure the constant growth of the edge weights,
we increase the weight an edge reduces from its neighborhood by a multiplicative factor.
This results in another deterioration in the approximation ratio, but has the benefit of forcing the weights of edges in the stack to exhibit a geometrical growth pattern.
This, in turn, creates the conditions for our modified local-ratio theorem to show its strength, allowing us to keep the size of the stack within the model's limits.
Carefully choosing parameters that manage the trade-off between space and approximation ratio,
we achieve a $(2+\epsilon)$-approximation using $O(n\log^2 n)$ bits.

Finally, we note that the basic structure of the local-ratio technique,
namely processing the edges one by one in an \emph{arbitrary} order and then performing some postprocessing,
suits very naturally to the streaming environment. Combined with the machinery we develop here in
order to follow the semi-streaming space constraints,
we believe this technique can be applied to additional problems in the semi-streaming model
and in similar computational models.

\paragraph{Related Work}
The study of graph algorithms in the semi-streaming model was initiated by Feigenbaum et al.~\cite{FeigenbaumKMSZ05},
in order to tackle the problem of processing massive graphs whose edge set cannot be stored in memory.
The need for algorithms for such massive graphs is evident,
as they become increasingly common:
graphs representing social networks, graphs for metabolic interactions used in computational biology
and even the communication graph of the Internet,
are only a few examples.

Feigenbaum et al.\ were also
the first to study the MWM problem in the semi-streaming model,
and presented a $6$-approximation algorithm for it.
Their algorithm maintains a matching at all times:
when an edge arrives, 
it checks if the new edge's weight
is more than double the sum of weights of its adjacent edges currently in the matching,
and if so,
the edge is added to the matching instead of its adjacent edges.
This idea was later adapted by McGregor~\cite{McGregor05} to achieve an approximation ratio of $5.828$,
by changing the threshold for inserting an edge to the matching
(McGregor also presents a $(2+\epsilon)$-approximation algorithm for the problem,
but using $O(\epsilon^{-3})$ passes on the input).
By using similar ideas,
while keeping deleted edges and reviving them later,
Zelke~\cite{Zelke12} achieves a $5.585$-approximation algorithm.

A different approach was taken
by Epstein et al.~\cite{EpsteinLMS11},
who achieve a $(4.911+\epsilon)$-approximation algorithm.
They use bucketing,
i.e., separate the edges into $O(\log n)$ weight classes,
find a matching in each bucket,
and then find the final matching in the union of these matching.
Crouch and Stubbs~\cite{CrouchS14} achieve an approximation
ratio of $(4+\epsilon)$
using related ideas,
but their algorithm uses weight classes
which are unbounded from above,
and thus are not disjoint.
Grigorescu et al.~\cite{GrigorescuMZ16}
have presented an improved analysis of the last algorithm,
claiming to achieve a $(3.5+\epsilon)$-approximation;
unfortunately, this analysis currently seems to contain an error.

The bucketing technique takes a heavy toll on the approximation factor,
and Crouch and Stubbs~\cite{CrouchS14}
prove this technique cannot give an approximation ratio
better than $3.5$.
To circumvent this bound,
we use a different
approximation framework,
the local-ratio technique.
To the best of our knowledge,
this is the first application of this technique in a streaming model.

Recently,
Ghaffari and Wajc~\cite{GhaffariW18} have shown a slight modification 
to our algorithm, 
that achieves the optimal $O(n\log n)$-bits bound,
assuming that the edge weights are integers of size polynomial in $n$.

One related problem is estimating size of the maximum matching in a graph~\cite{AssadiKL17,Kapralov13,KapralovKS14,GoelKK12}
which is known to be related to matrix rank approximation.
More general submodular-function matching problems
in the semi-streaming model have been considered by
Varadaraja\cite{Varadaraja11} and by Chakrabarti and Kale~\cite{ChakrabartiK14}.

The MWM problem was also considered in other streaming models,
such as the MapReduce model \cite{CrouchS14, LattanziMSV11},
the sliding-window model \cite{CrouchS14, CrouchMS13} and the turnstile stream model (allowing deletions as well as insertions)~\cite{Konrad15, AssadiKLY16, BuryS15, ChitnisCEHMMV16}.
Extending our technique to other computational models
is a challenge yet to be addressed.

\paragraph{Structure of this Paper}
We formally define the MWM problem and the semi-streaming model of computation 
in Section~\ref{sec: preliminaries}.
In Section~\ref{sec: approximaton} we introduce the local-ratio theorem,
present a sequential 2-approximation local-ratio algorithm for MWM,
and discuss our variations to the theorem.
In Section~\ref{sec: semi streaming algo}
we extend the 2-approximation algorithm to a more involved
$(2+\epsilon)$-approximation algorithm for MWM,
analyze its performance,
and finally adapt it to the semi-streaming model.

\section{Preliminaries}
\label{sec: preliminaries}

Let $G=(V,E,w)$ be a simple graph with non-negative edge weights,
$w\in \bbR_{+}^E$
(we use vector notation for edge weights).
Denote $n=\size{V}$ and $m=\size{E}$;
for an edge $e$ denote
$N(e) = \set{e'\mid \size{e\cap e'}=1}$
called the \emph{neighboring edges} of $e$,
and $N^+(e) = N(e)\cup \set{e}$.
We usually assume edge weights 
and their sums can be represented by $O(\log n)$ bits,
and discuss other weight functions at the end of the paper.

\paragraph{Maximum Weight Matching} 
A \emph{matching} in $G$ is a set $M\subseteq E$ of edges
such that no two edges share a node. 
A \emph{maximum weight matching} (MWM) in $G$
is a matching $M$ of maximum weight:
for every matching $M'$ in $G$,
we have $\sum_{e\in M} w[e] \geq \sum_{e\in M'} w[e]$.

A matching $M$ is identified with its indicator vector $x$,
defined by $x[e]=1$ if $e\in M$,
and $x[e]=0$ otherwise.
Thus, the weight of a matching $x$ is the value of the
inner product of $x$ and $w$, denoted $xw$.
A set of feasibility constrains on $x$ is induced
by the graph
in a straightforward manner:
	$\forall e,e'\in E: \size{e\cap e'}=1\implies x[e]\cdot x[e']=0$.

\paragraph{Approximation Algorithms}
A feasible matching $x$ is said to be
a \emph{$p$-approximation} of a MWM in $G$,
for a constant $p\geq 1$,
if every matching $x^\ast$ satisfies
$x^\ast w \leq p\cdot xw$.
An algorithm returning a $p$-approximation
on every input graph
is said to be a \emph{$p$-approximation algorithm} for
the MWM problem,
and $p$ is called the \emph{approximation ratio}
of the algorithm.
Note that if $p'>p$ than a $p$-approximation algorithm
is also a $p'$-approximation algorithm.
The definition naturally extends to other optimization problems.

\paragraph{The Semi-Streaming Model}
In the semi-streaming model of computation, as in sequential models,
the goal is to compute parameters in some given graph. An algorithm in this model proceeds in iterations,
where in each iteration it receives an edge from the stream
and processes it.
Since the number of edges in the graph
might be too large to fit in memory,
we limit the algorithm to use only $O(n \polylog n)$ bits.
In addition,
we try to keep the processing times of the edges
as short as possible,
since a long processing time might result in 
a queue of later incoming edges, exceeding the space limitations.
The algorithm is also allowed to perform pre-processing and post-processing, 
but minimizing their times is of less importance.

\section{Approximating Maximum Weight Matching}
\label{sec: approximaton}

In this section we present the \emph{local-ratio} theorem for maximization problems~\cite{Bar-YehudaBFR04},
and use it to present a sequential 2-approximation algorithm for MWM.
We then present extensions of this technique
and use them
in order to adjust the sequential local-ratio algorithm
to the semi-streaming model,
incurring only a small loss in the approximation ratio.

\subsection{A Simple Local-Ratio Approximation Algorithm for MWM}
The basic building blocks of a local-ratio algorithm are
iterative \emph{weight reduction} steps.
Weight reduction step number $i$ starts with a graph 
$G=(V,E,w_i)$,
and defines two new graphs,
composed of $(V,E)$ and new edge-weight functions:
the \emph{reduced graph}, which has weight function $w_{i+1}$,
and the \emph{residual graph}, which has weight function $\wbar_{i+1}$,
such that $w_i=w_{i+1}+\wbar_{i+1}$.
We start with the local-ratio theorem
for maximization problems~\cite[Theorem 9]{Bar-YehudaBFR04},
which we restate here for completeness.
Note that this theorem applies
even if $w_{i+1}$ takes negative values.

\begin{theorem}
\label{thm: lr original}
Let $w_i \in \bbR^m$ be a vector,
and consider the problem of maximizing the product $xw_i$
under a set of feasibility constraints.
Let $w_{i+1},\wbar_{i+1} \in \bbR^m$ be vectors such that $w_i=w_{i+1}+\wbar_{i+1}$.
If $x_i\in \bbR^m$ is a feasible solution
that is a $p$-approximation with respect to $w_{i+1}$
and with respect to $\wbar_{i+1}$,
then $x_i$ is a $p$-approximation with respect to $w_i$ as well.
\end{theorem}

\begin{proof}
Let $x_i^\ast, x_{i+1}^\ast$ and $\xbar_{i+1}^\ast$
be maximum feasible solutions with respect to
$w_i, w_{i+1}$ and $\wbar_{i+1}$. Then
\begin{align*}
    x_i^\ast w_i
       	&=      x_i^\ast w_{i+1} + x_i^\ast \wbar_{i+1}\\
    	&\leq   x^\ast_{i+1} w_{i+1} + \xbar^\ast_{i+1} \wbar_{i+1}\\
    	&\leq   p \cdot x_i w_{i+1} + p \cdot x_i \wbar_{i+1}\\
    	&=     p \cdot x_i w_i,
\end{align*}
where the first inequality follows from the maximality of $x^\ast_{i+1}$ and $\xbar^\ast_{i+1}$,
and the second from the assumption that $x_i$ is a $p$-approximation
with respect to $w_{i+1}$ and $\wbar_{i+1}$.
\end{proof}

We apply weight reduction steps iteratively,
while ensuring that any $p$-approximate solution to $w_{i+1}$
can be easily extended into a $p$-approximate solution
to $\wbar_{i+1}$.

For the specific problem of MWM,
a weight reduction step is done by picking an arbitrary
edge $e\in E$ of positive weight and
reducing this weight from every $e' \in N^+(e)$.
This splits the weight vector $w_i$ into two vectors,
$w_{i+1}$ and $\wbar_{i+1}$, 
by setting
\begin{align*}
\wbar_{i+1}[e'] =
\begin{cases}
w_i[e] & e'\in N^+(e); \\
0 &	 	\text{otherwise,}
\end{cases}
\end{align*}
and $w_{i+1} = w_i - \wbar_{i+1}$.
Any $2$-approximate solution $x_{i+1}$ for the reduced graph
can be easily extended into a $2$-approximate solution 
for the residual graph
by making sure that at least one edge
$e' \in N^+(e)$ is in the solution:
if this is not the case,
we can add $e$ to the solution
without violating the constraints.
As $w_{i+1}[e]=0$, adding $e$ to the solution
does not reduce the solution's value with respect to $w_{i+1}$.
Thus, we get a 2-approximate solution
for both $w_{i+1}$ and $\wbar_{i+1}$.

This simple technique is realized by Algorithm~\ref{alg: MWM-seq}.
First, it applies weight reduction steps iteratively
using edges of positive reduced weight,
splitting a weight function $w_i$ into
$w_{i+1}$ (reduced) and $\wbar_{i+1}$ (residual)
and keeping the edge in a stack.
When no edge with a positive reduced weight remains,
the algorithm unwinds the stack
and adds the edges greedily to the matching.
While unwinding the stack, it maintains a set of interim solutions $\set{x_i}$;
the local-ratio theorem guarantees 
that every $x_i$ is a 2-approximate solution for $w_i$. 
Finally, the algorithm returns $x_1$, 
which is a 2-approximate solution for the original problem.

We note that this algorithm does not work in the semi-streaming model, as the stack can easily grow to contain $\Omega(n^2)$ edges.

\RestyleAlgo{boxruled}
\LinesNumbered
\begin{algorithm}[t]
	\DontPrintSemicolon
	\caption{\texttt{MWM-simple}($V,E,w$). A simple 2-approximation algorithm for MWM}
	\label{alg: MWM-seq}
	$S \gets $ empty stack \\
	$w_1 \gets w$;
	$i\gets 1$\\
	\ForEach{$e_i\in E$ s.t. $w_i[e_i] \geq 0$}
	{
		$S.\text{push}(e_i)$ \\
		$w_{i+1} \gets w_i$ \\
		\lForEach(\tcp*[f]%
		{Implicit: $\wbar_{i+1}[e'] \gets w_i[e_i]$})%
		{$e' \in N^+(e_i)$}
		{$w_{i+1}[e'] \gets w_i[e']- w_i[e_i]$}
		$i \gets i+1$\\
	}
	
	$k\gets \size{S}$ \\
	$x_{k+1} \gets  \vec{0}$\\
	\For{$i\gets k$ down to $1$}
	{
		$x_{i} \gets x_{i+1}$ \\
		$e_i \gets  S[i]$\\
		\lIf{$\forall e\in N^+(e_i):$ $x_{i}[e] = 0$}
		{
			$x_{i}[e_i] \gets 1$	
		}
	}
	return $x_1$
\end{algorithm}

\subsection{Extending the Local-Ratio Technique}
We now extend the approximation techniques used
in Algorithm~\ref{alg: MWM-seq}.
This allows us to present another sequential approximation algorithm
for MWM
in the following section,
with a worse approximation ratio of $2+\epsilon$.
However,
from the new algorithm we derive the desired
approximation algorithm
for the semi-streaming model,
with no further increase in the approximation ratio.

If instead of reducing exactly $w_i[e]$
from the neighboring edges of $e$,
we reduce either $w_i[e]$ or $\alpha w_i[e]$ from each such edge,
for some $\alpha\geq 1$,
we get a $(2\alpha)$-approximation,
as formalized by the next lemma.

\begin{lemma}
\label{lem: residual is an approximate sol}
Let $w_i, w_{i+1}$ and $\wbar_{i+1}$ be weight functions
and $e\in E$ an edge such that
\begin{align}
\wbar_{i+1}[e'] =
\begin{cases}
w_i[e] & e'=e; \\
\alpha w[e] \text{ or } w[e] & e'\in N(e); \\
0 & \text{otherwise,}
\end{cases}
\end{align}
and $w_{i+1}=w-\wbar_{i+1}$;
the choice between $w_i[e]$ and $\alpha w_i[e]$
can be arbitrary.

Let $x_i\in \{0,1\}^m$ be a matching.
If $x_i[e'] \neq 0$ for some $e'\in N^+(e)$,
then $x_i$ is a $(2\alpha)$-approximate solution for $\wbar_{i+1}$.
\end{lemma}

\begin{proof}
Let $x_i^\ast$ be any matching.
The definition of $\wbar_{i+1}$ guarantees that $x_i^\ast$ contains
at most two edges of non-zero weight in $\wbar_{i+1}$,
each of weight at most $\alpha w_i[e]$,
so $x^\ast \wbar_{i+1} \leq 2\alpha w_i[e]$.
On the other hand,
$x_i[e'] \neq 0$ for some $e'\in N^+(e)$,
so $w_i[e] \leq x_i \wbar_{i+1}$.
Using the last two inequalities, we get 
$x^\ast \wbar_{i+1} \leq 2\alpha\cdot x_i \wbar_{i+1}$,
as desired.
\end{proof}

Next,
we note that if the optimal solution for the reduced graph
is greater than the optimal solution for the residual graph
by some multiplicative factor $p \geq 1$,
then it is also a $(1+1 / p)$-approximation for the original graph.
For large values of $p$,
an approximate solution for the reduced graph
gives roughly the same approximation ratio for the original graph,
which allows us to ignore the residual graph.
We formalize this in the next lemma.

\begin{lemma}
\label{lem: lr with omissions}
Let $w_i, w_{i+1}$ and $\wbar_{i+1}$ be weight functions satisfying
$w_i= w_{i+1} + \wbar_{i+1}$ and
$w_{i+1}[e] \leq w_i[e]$ for all $e\in E$.
Let $x_{i+1}$ be a $\beta$-approximate solution for $w_{i+1}$.

If $x_{i+1} w_{i+1}$  is at least $p$ times larger than
any matching in $\wbar_{i+1}$,
then $x_{i+1}$ is a $(\beta+1/p)$-approximate solution
for $w_i$.
\end{lemma}

\begin{proof}
Let $x_i^\ast, x_{i+1}^\ast$ and $\xbar_{i+1}^\ast$
be matchings of maximum weights in
$w_i, w_{i+1}$ and $\wbar_{i+1}$ respectively.

The assumptions imply
$x_{i+1}^\ast w_{i+1} \leq \beta x_{i+1} w_{i+1}$ and
$p \xbar_{i+1}^\ast \wbar_{i+1} \leq x_{i+1} w_{i+1}$,
so
\begin{align*}
	x_i^\ast w_i
	&=     x_i^\ast w_{i+1} + x_i^\ast \wbar_{i+1}\\
	&\leq  x^\ast_{i+1} w_{i+1} + \xbar^\ast_{i+1} \wbar_{i+1}\\
	&\leq  \beta x_{i+1}  w_{i+1} + (1/p) x_{i+1} w_{i+1}\\
	&=     (\beta+1/p) x_{i+1} w_{i+1}\\
	&\leq  (\beta+1/p) x_{i+1} w_i,
\end{align*}
where the last inequality follows from the fact that
$w_{i+1}[e] \leq w_i[e]$ for all $e\in E$.
\end{proof}

Let $w_1$ be a weight vector for the MWM problem,
and consider an iterative splitting of $w_i$ into
$w_{i+1}$ and $\wbar_{i+1}$ for $k$ times.
The last lemma allowed us to ignore the residual graph once;
we now extend it to allow the iterative omission of the residual graph.

Denote $\alpha = \sqrt{1+\epsilon/2}$,
$\gamma = n^2 / \ln (\alpha)$,
and $\beta_i = 2\alpha(1+1/\gamma)^{k+1-i}$ for all $i$.

\begin{lemma}
\label{lem: approx-ratio}
Let $G=(V,E,w_1)$ a graph,
and $w_2,\ldots,w_{k+1}$ and $\wbar_2,\ldots \wbar_{k+1}$
sequences of reduced and residual weight functions
for $(V,E)$, respectively.

Assume that we generate a sequence of solutions
$x_{k+1},\ldots, x_1$,
such that $x_{k+1}$ is an optimal solution for $w_{k+1}$,
and that for $1\leq i \leq k$,
if $x_{i+1}$ is a $\beta_{i+1}$-approximate solution
for $w_{i+1}$
then $x_i$ has the following properties:
\begin{enumerate}
  \item
  $x_i$ is a $\beta_{i+1}$-approximate solution for $w_{i+1}$.
  \item
  At least one of the following holds:
  \begin{enumerate}
    \item $x_i$ is a $\beta_{i+1}$-approximate solution for $\wbar_{i+1}$; or
    \item $x_{i} w_{i+1} \geq (\gamma / \beta_{i+1}) x^\ast\wbar_{i+1} $
        for every solution $x^\ast$.
  \end{enumerate}
\end{enumerate}
Then
$x_1$ is a $\beta_1$-approximate solution for $w_1$.
\end{lemma}

\begin{proof}
We prove,
by induction on $i$ ranging from $k+1$ down to $1$,
that $x_i$ is a $\beta_i$-approximate solution for $w_i$.

The base, $i=k+1$, is trivial by the assumption on $x_{k+1}$.

Assume the claim is true for $x_{i+1}$,
then condition $1$ holds for $x_i$.
If condition $2(a)$ holds,
then by condition $1$
and the local-ratio theorem
(Theorem~\ref{thm: lr original}),
$x_i$ is a $\beta_{i+1}$-approximate solution for $w_i$.
Because $\beta_i > \beta_{i+1}$,
$x_i$ is also a $\beta_i$-approximate solution for $w_i$.
If condition $2(b)$ holds,
then from condition $1$ and
Lemma~\ref{lem: lr with omissions} we deduce that
$x_{i}$ is a $(\beta_{i+1}+\beta_{i+1}/\gamma)$-%
approximate solution for $w_i$.
The definition of $\beta_i$ yields:
\begin{align*}
	\beta_{i+1}+\beta_{i+1}/\gamma
	&=  (1+1/\gamma)\cdot 2\alpha(1+1/\gamma)^{k+1-(i+1)}\\
	&=  2\alpha(1+1/\gamma)^{k+1-i}
	=  \beta_i.
\end{align*}
Specifically,
$x_1$ is a $\beta_1$-approximate solution for $w_1$,
and the proof is complete.
\end{proof}

\section{A Semi-Streaming Algorithm}
\label{sec: semi streaming algo}
We present a $(2+\epsilon)$-approximation algorithm
for the MWM problem using our extension of the local-ratio technique.
This algorithm is suitable for a streaming model which has no space constraints,
but not for the semi-streaming model.
We then present a lightweight variant of the algorithm,
which obeys the space constraints of the semi-streaming model.

The new algorithm is similar to Algorithm~\ref{alg: MWM-seq}:
it performs a series of weight reduction steps
defining a series of reduced weight functions $\set{w_i}$,
and then constructs a series of approximate solutions $\set{x_i}$.
To prove the desired approximation ratio is achieved,
we use Lemma~\ref{lem: approx-ratio}
as a substitute for the local-ratio theorem.

We start by presenting the challenges
posed by the semi-streaming model,
and the ways in which the new algorithm deals with them.
Let $e_i$ be the edge considered in iteration $i$.

\paragraph{Retroactive weight reduction}
The sequential algorithm constructs $w_{i+1}$ from $w_i$
using an edge $e_i$,
by reducing $w_i(e_i)$ form  the weight of every $e' \in N^+(e_i)$.
This cannot be done directly in the semi-streaming model,
as some edges of $N^+(e_i)$ might only arrive after $e_i$ is processed.
Instead, the algorithm keeps a variable $\phi_i(v) = \sum_{j=1}^i w_j[e_j]$
for every node $v\in V$.
When a new edge $e=(u,u')$ arrives,
its reduced weight is first computed,
by reducing $\phi_{i-1}(u)$ and $\phi_{i-1}(u')$ from its original weight.

\paragraph{Removing edges from the stack}
In the sequential algorithm,
the stack may grow to hold all of the graph edges.
Lemma~\ref{lem: approx-ratio} presents conditions
under which an approximate solution for $w_{i+1}$
is also an approximate solution for $w_{i}$.
When these conditions are met,
we may remove the edge $e_i$ from the stack,
which we use in order to make sure that
the stack's size does not exceed $O(n \log n)$ edges.

\paragraph{Assuring edge-weight growth}
In order to make sure edges are removed from the stack,
we force a small but consistent growth in the edge weights
around each node.
Roughly speaking,
the edge weights grow geometrically by a multiplicative
$\alpha$ factor;
after a logarithmic number of new edges considered,
the weights grow large enough to allow the algorithm
to neglect the older edges and remove them from the stack.

\subsection{\lrmwm{}}

\lrmwm{} (Algorithm~\ref{alg: MWM}) has two phases:
in the first phase,
it iterates over the edges
and pushes chosen edges into a stack.
In the second phase,
the edges are popped out of the stack
and added greedily to the matching.

The algorithm begins with an edge-weight function $w_1$,
given as input.
For each node $v$,
the algorithm explicitly maintains a non-negative weight function $\phi_i(v)$,
which is used to filter edges (Line~\ref{line: filter light edges}):
an edge $e=(u,u')$ processed at iteration $i$ is \emph{light}
if $w_1[e] \leq \alpha(\phi_{i-1}(u) + \phi_{i-1}(u'))$,
and \emph{heavy} otherwise.
In iteration $i$,
the algorithm processes the incoming edges,
ignoring light edges until a heavy edge is encountered.
This edge, denoted by $e_i$, is used to update $w_{i}$ and $\phi_i(v)$.
Eventually,
all heavy edges are denoted with sub-indexes ($e_i$),
while the light edges are left un-tagged ($e$).

When an edge $e=(u,u')$ is processed in iteration $i$,
the algorithm performs all weight reduction steps on $e$
retroactively using $\phi_{i-1}(u)$ and $\phi_{i-1}(u')$,
to set the value of $w_i[e]$.
It decides between reducing $\phi_{i-1}(u) + \phi_{i-1}(u')$ or $\alpha(\phi_{i-1}(u) + \phi_{i-1}(u'))$ from the weight of $e$,
in a way that guarantees a geometric growth of $\phi$,
implying a bound on the size of the stack.

For every node $v$, we hold a queue $E_i(v)$. 
This is a list of the heavy edges containing the node $v$ 
currently present in the stack. 
Upon the arrival of a heavy edge $e_i=(u,u')$, 
we perform a weight reduction step:
$\phi_{i-1}(u)$ and $\phi_{i-1}(u')$ are increased by $w_i[e_i]$,
and $e_i$ is pushed into the stack. 
We also enqueue $e_i$ in $E_i(u)$ and $E_i(u')$.
If the size of $E_i(u)$ or $E_i(u')$ exceeds a certain bound, 
we dequeue an edge from the exceeding queue, 
and remove it from the stack.
For the sake of analysis, 
we do not remove edges from the stack in Algorithm~\ref{alg: MWM}, 
but only replace them by a $\bot$ sign.

In the second phase, the algorithm unwinds the stack,
adding edges greedily to the matching while ignoring $\bot$ symbols.
The usage of the $\bot$ symbol is replaced by 
deletion of the relevant edge
in the semi-streaming algorithm, presented in the next subsection.

\RestyleAlgo{boxruled}
\LinesNumbered
\begin{algorithm}[t]
	\DontPrintSemicolon
	\caption{\texttt{MWM-seq}($V,E,w$). A sequential approximation algorithm for MWM}
	\label{alg: MWM}
	$S \gets $ empty stack \\
	$w_1 \gets w$;
	$\phi_0\gets \vec{0}$;
	$c_0\gets \vec{0}$ \,  \tcp*{$c_i$ is only used for the proof}
	$\forall v\in V: E_0(v) \gets $ empty queue\\
	$i\gets 1$\\
	
	\ForEach{$e=(u,u')\in E$}
	{
		\lIf(\tcp*[f]%
		{Implicit: $\wbar_{j+1}[e] \gets \alpha  w_j[e_j]$ for every $e_j\in N(e)$})%
		{$w_1[e] \leq \alpha(\phi_{i-1}(u) + \phi_{i-1}(u'))$}%
		{continue}%
		\label{line: filter light edges}
		$e_i \gets e$ \label{line: set ei}\\
		$S.\text{push}(e_i)$\\
		$w_i[e_i]\gets w_1[e_i]-
		(\phi_{i-1}(u) + \phi_{i-1}(u'))$ 
		\tcp*{Implicit: $\wbar_{j+1}[e_i] \gets w_j[e_j]$
			\label{line: set ei's weight}
			for every $e_j\in N^+(e_i)$}
		$\phi_{i} \gets \phi_{i-1}$; $E_{i} \gets E_{i-1}$;
		$c_{i} \gets c_{i-1}$\\
		\ForEach{$v \in e_i$}
		{
			$c_i(v) \gets c_i(v) + 1$ \\
			$E_i(v).\text{enqueue}(e_i)$ \\
			$\phi_i(v) \gets \phi_{i-1}(v) + w_i[e_i]$
			\label{line: update phiiv}\\
			\If{$ (\alpha-1)\alpha^{\size{E_i(v)}-2} >  2\alpha\gamma$%
				\label{line: omit obsolete edge}}
			{
				$e_j\gets E_i(v).\text{dequeue}()$
				\label{line: dequeue}\\
				$S[j] \gets \bot$ \\
			}
		}
		$i \gets i+1$\\
	}
	$k\gets \size{S}$ \\
	$x_{k+1} \gets  \vec{0}$\\
	\For{$i\gets k$ down to $1$}
	{
		$x_{i} \gets x_{i+1}$ \\
		$e_i \gets  S[i]$\\
		\lIf{$e_i = \bot$}%
		{continue}%
		\label{line: edge deleted cond}
		\lIf{$\forall e\in N(e_i):$ $x_{i}[e] = 0$}%
		{$x_{i}[e_i] \gets 1$}%
		\label{line: add edge cond}
	}
	return $x_1$
\end{algorithm}

We start the analysis of \lrmwm{} by proving that the
node-weight functions $\phi_i(v)$ grow geometrically with $i$.
In the algorithm,
the variable $c_i(v)$ counts the heavy edges containing $v$ that
arrive until iteration $i$.
Its value is not used in the algorithm itself;
we only use it in the proof,
to bound from below the growth $\phi(v)$.
In various places in the analysis we consider the expression $c_j(v) - c_i(v)$,
which is the number of heavy edges added to $v$ from iteration $i$ until iteration $j$. We eventually show that the reduced weights of heavy edges exhibit a growth pattern exponential in $c_j(v) - c_i(v)$.

\begin{lemma}
\label{lem: phi grows}
For every $v \in V$ and $j\geq i$, 
$\phi_j(v) \geq \alpha^{c_j(v) - c_i(v)}\phi_i(v)$.
\end{lemma}

\begin{proof}
We fix $i$
and prove the lemma by induction on $j$, where $j\geq i$.
The base case, $j=i$, is trivial.

For $j>i$, we consider two cases:
if $v \notin e_j$ then $c_{j}(v) = c_{j-1}(v)$,
so $\phi_{j}(v) = \phi_{j-1}(v) \geq \alpha^{c_{j-1}(v) - c_i(v)}
= \alpha^{c_{j}(v) - c_i(v)}$
by the induction hypothesis.

Otherwise, $e_j=(v,u)$ for some $u\in V$,
and
\begin{align*}
\phi_{j}(v)
&=     \phi_{j-1}(v) + w_j[e_j]
\tag{Line~\ref{line: update phiiv}} \\
&\geq  \phi_{j-1}(v) + (\alpha-1)(\phi_{j-1}(v) + \phi_{j-1}(u))
\tag{Line~\ref{line: filter light edges}} \\
&\geq  \phi_{j-1}(v) + (\alpha-1)\phi_{j-1}(v)
\tag{Line~\ref{line: set ei's weight}} \\
&=     \alpha \phi_{j-1}(v)\\
&\geq  \alpha\cdot \alpha^{c_{j-1}(v) - c_i(v)} \phi_i(v)
\tag{induction hypothesis} \\
&=      \alpha^{c_{j}(v) - c_i(v)} \phi_i(v),
\tag{$v\in e_j$ implies $c_{j}(v) = c_{j-1}(v) +1$} 
\end{align*}
as desired.
\end{proof}

Consider the sequences of reduced and residual
edge-weight functions,
$w_2, \ldots, w_{k+1}$ and $\wbar_2, \ldots, \wbar_{k+1}$,
induced by the algorithm.
Note that these weight functions are defined on the fly:
when an edge $e$ arrives, 
it implicitly sets $w_{j+1}(e)$ and $\wbar_{j+1}(e)$
of each of its adjacent heavy edges $e_j$,
where $e_j$ may arrive before or after $e$.
Thus, the values of $w_i$ and $\wbar_{i}$ 
are completely determined only when the first phase ends,
and so does the length $k$ of the sequences.

The weight functions are defined inductively as follows.
We formally define $w_1 = w$,
where $w$ is the function given as input.
The edge $e_i$ is used to split
the weight function $w_i$ into $w_{i+1}$ and $\wbar_{i+1}$,
the latter defined by
\begin{align}
\wbar_{i+1}[e'] =
\begin{cases}
w_i[e_i]         & e'=e_i; \\
w_i[e_i]         & e'\in N(e_i) \text{ and $e'$ is heavy}; \\
\alpha w_i[e_i]  & e'\in N(e_i) \text{ and $e'$ is light}; \\
0                & \text{otherwise.}
\end{cases}
\end{align}
and the former by
$w_{i+1} = w_i - \wbar_{i+1}$.
The length $k$ is the number of heavy edges
encountered in the first phase.
Note that $\wbar$ is non-negative,
so $w_i[e]$ is a non-increasing function of $i$,
for any fixed edge $e$.

The next lemma focuses on a node $v$ and two heavy edges adjacent to it,
$e_i$ and $e_j$.
It asserts that for $j>i$, the reduced weight at iteration $i+1$ of a heavy edge $e_j$ grows geometrically with respect to $w_i[e_i]$.

\begin{lemma}
\label{lem: edge weight inequality}
Let $e_i, e_j \in E$
such that $j > i$
and $e_i\cap e_j  = \set{v}$.
Then $w_{i+1}[e_j]
>(\alpha-1) \alpha^{c_j(v)-c_i(v)-1} w_{i}[e_i]$.
\end{lemma}

\begin{proof}
The lemma follows by a simple computation.
As $w_j[e]$ is a non-increasing
\begin{align*}
w_{i+1}[e_j]
& \geq w_{j}[e_j]
     \\
& \geq (\alpha-1) \phi_{j-1} (v)
    \tag{Lines~\ref{line: filter light edges}
         and~\ref{line: set ei's weight}} \\
& \geq (\alpha-1) \alpha^{c_{j-1}(v) - c_i(v)} \phi_i(v)
    \tag{Lemma~\ref{lem: phi grows}} \\
& \geq (\alpha-1) \alpha^{c_{j-1}(v) - c_i(v)} w_i[e_i]
    \tag{Line~\ref{line: update phiiv}} \\
& =    (\alpha-1) \alpha^{c_{j}(v) - c_i(v) - 1} w_i[e_i]
    \tag{$v\in e_j$ implies $c_{j-1}(v) = c_j(v) -1$}
\end{align*}
as desired.
\end{proof}

In the second loop of the algorithm,
the edges are taken out of the stack
and a solution is greedily constructed.
The algorithm's approximation ratio is
the approximation ratio of the solution $x_1$
on the original weight function $w_1$.
To bound this quantity,
we prove by induction that every $x_i$ is
a $\beta_i$-approximate solution for $w_i$.
We break our analysis into cases,
for which we need the next three lemmas.
First, we consider an edge $e_i$ which is evicted from the stack.

\begin{lemma}
	\label{lem: xi case 1}
If $x_{i+1}$ is a $\beta_{i+1}$-approximate solution for $w_{i+1}$
and the condition in Line~\ref{line: edge deleted cond} holds
for $e_i$,
then $x_i$ is a $\beta_i$-approximate solution for $w_i$.
\end{lemma}
\begin{proof}
Since the condition in Line~\ref{line: edge deleted cond} holds,
we have $x_i= x_{i+1}$.
This immediately guarantees that $x_i$ is a feasible solution and
that condition $1$ of Lemma~\ref{lem: approx-ratio} holds.
We show that condition $2(b)$  of Lemma~\ref{lem: approx-ratio}
holds as well.

As the condition in Line~\ref{line: edge deleted cond} holds,
we know that in some iteration $j$ of the first phase, $j>i$,
the condition in Line~\ref{line: omit obsolete edge} held.
That is,
for some endpoint $v$ of $e_i$, an edge $e_{j}$ with $e_i\cap e_{j}=\set{v}$ was enqueued into $E_{j}(v)$,
the condition
$(\alpha-1)\alpha^{\size{E_{j}(v)}-2} >  2\alpha\gamma$ held,
and $e_i$ was then dequeued from $E_{j}(v)$.

Every enqueue operation to $E_i(v)$ is accompanied by
an increases of $c_i(v)$ by $1$,
so when the condition in Line~\ref{line: omit obsolete edge}
was checked, $e_i$ and $e_{j}$ were the oldest and newest elements in $E_{j}(v)$, respectively,
and the size of $E_{j}(v)$ was exactly $c_{j}(v) - c_i(v) +1$.
Thus,
$(\alpha-1)\alpha^{c_{j}(v)-c_{i}(v) -1} \geq 2\alpha\gamma$.

Using this inequality and Lemma~\ref{lem: edge weight inequality},
we have
$$w_{i+1}[e_{j}]
\geq (\alpha-1)\alpha^{c_{j}(v)-c_{i}(v) -1} w_{i}[e_{i}]
\geq 2\alpha\gamma w_{i}[e_{i}].$$
Hence, the single edge $e_{j}$
is a matching of weight at least $2\alpha\gamma w_{i}[e_i]$
in $w_{i+1}$.
As $x_{i+1}$ is a $\beta_{i+1}$-approximate solution for $w_{i+1}$,
we have $\beta_{i+1} x_{i+1}w_{i+1} \geq 2\alpha\gamma w_i[e_i]$.

The definition of $\wbar_{i+1}$ guarantees it has the following structure:
\begin{align}
\wbar_{i+1}[e'] =
\begin{cases}
w_i[e_i] & e'=e_i; \\
\alpha w_i[e_i]\text{ or } w_i[e_i] & e'\in N(e_i); \\
0 & \text{otherwise.}
\end{cases}
\end{align}
Thus, any solution $x^\ast$ for $\wbar_{i+1}$
contains at most two edges,
of weight at most $\alpha w_i[e_i]$,
i.e.\ $2\alpha w_i[e_i] \geq x^\ast \wbar_{i+1}$.
The last two inequalities guarantee any solution
$x^\ast$ satisfies
$$   (\beta_{i+1} / \gamma)x_{i+1}w_{i+1}
\geq 2\alpha w_i[e_i]
\geq x^\ast \wbar_{i+1} $$
so
$ x_{i}w_{i+1} = x_{i+1}w_{i+1}
\geq (\gamma/\beta_{i+1})x^\ast \wbar_{i+1} $,
and condition $2(b)$ of Lemma~\ref{lem: approx-ratio} holds.
\end{proof}

We now turn to the case of an edge $e_i$ that is not evicted from the stack.
\begin{lemma}
	\label{lem: xi case 2}
If $x_{i+1}$ is a $\beta_{i+1}$-approximation for $w_{i+1}$
and the condition on Line~\ref{line: edge deleted cond} does not hold for $e_i$,
then $x_i$ is a $\beta_i$-approximation for $w_i$.
\end{lemma}
\begin{proof}
If the condition on Line~\ref{line: add edge cond} holds,
then $x_i$ is derived from $x_{i+1}$
by adding $e_i$ to $x_{i+1}$.
The condition in this line,
together with the assumption that $x_{i+1}$ is a matching,
guarantee that $x_i$ is a matching.
Since $\wbar_{i+1} [e_i] = w_i[e_i]$
and $w_{i+1} = w_i - \wbar_{i+1}$,
we have $w_{i+1} [e_i] = 0$.
Hence, $x_i w_{i+1} = x_{i+1} w_{i+1}$,
so $x_i$ is also a $\beta_{i+1}$-approximate solution
for $w_{i+1}$
and condition~$1$ of Lemma~\ref{lem: approx-ratio} holds.
By Lemma~\ref{lem: residual is an approximate sol},
$x_i$ is a $(2\alpha)$-approximate solution for $\wbar_{i+1}$,
and because $2\alpha \leq \beta_{i+1}$ it is also
a $\beta_{i+1}$-approximate solution to $\wbar_{i+1}$
and condition~$2(a)$ of Lemma~\ref{lem: approx-ratio} holds.

Finally, if the condition in Line~\ref{line: add edge cond} does not hold, we set $x_i = x_{i+1}$.
Then $x_i$ is a feasible matching
satisfying condition~$1$ of Lemma~\ref{lem: approx-ratio}.
The condition in Line~\ref{line: add edge cond} does not hold,
so $x_{i+1}[e']\neq 0$ for some $e'\in N^+[e_i]$,
and Lemma~\ref{lem: residual is an approximate sol} promises
$x_i$ is a $(2\alpha)$-approximation for $\wbar_{i+1}$.
As before,
$2\alpha \leq \beta_{i+1}$ proves that
condition~$2(a)$ of Lemma~\ref{lem: approx-ratio} holds.
\end{proof}

Finally, we show that when the first phase ends,
none of the reduced edge weights is positive.

\begin{lemma}
\label{lem: last w is non-positive}
At the end of the first phase,
$w_{k+1}[e] \leq 0$ for all $e\in E$.
\end{lemma}

\begin{proof}
Consider an edge $e$.
If $e=e_i$ is heavy then
$\wbar_{i+1} [e_i] = w_i[e_i]$ and $w_{i+1} = w_i - \wbar_{i+1}$
imply $w_{i+1} [e_i] = 0$.
The monotonicity of $w_i[e]$ completes the proof.

If $e=(u,u')$ is a light edge considered in iteration $i$,
then $w_1[e] \leq \alpha(\phi_{i-1}(u) + \phi_{i-1}(u'))$.
Line~\ref{line: update phiiv} guarantees
$$\phi_{i-1}(u) =
\sum_{\set{e_j\left|\substack{ u\in e_j\\ j\leq i-1} \right.}}
w_j[e_j],$$
and a similar claim holds for $u'$.
On the other hand,
$w_{j+1} = w_j - \wbar_{j+1}$ and
$\wbar_{j+1}[e] = \alpha w_j[e_j]$ for all $e_j\in N(e)$.
Hence $w_{j+1}[e] = w_j[e]-\alpha w_j[e_j]$,
and a simple induction implies
$$w_i[e]
=  w_1[e] -
\alpha \sum_{\set{e_j\left|\substack{e_j\in N(e)\\
j\leq i-1} \right.}}
w_j[e_j].$$
The last two equalities, together with the definition of $N(e)$,
imply
$w_i[e] = w_1[e] - \alpha (\phi_{i-1}(u) + \phi_{i-1}(u'))$.
The inequality $w_1[e] \leq \alpha(\phi_{i-1}(u) + \phi_{i-1}(u'))$
implies $w_i[e]\leq 0$ for all $e\in E$,
and the monotonicity of $w_i[e]$
completes the proof.
\end{proof}

We are now ready to prove the main theorem of this section.
\begin{theorem}
\lrmwm{} returns a $(2+\epsilon)$-approximation
for the MWM problem.
\end{theorem}

\begin{proof}
By Lemma~\ref{lem: last w is non-positive},
the first loop ends when $w_{k+1} \leq \vec{0}$,
so $x_{k+1}= \vec{0}$ is indeed an optimal solution for $w_{k+1}$.

Assume $x_{i+1}$ is a $\beta_{i+1}$-approximate solution for $w_{i+1}$.
From Lemmas~\ref{lem: xi case 1} and \ref{lem: xi case 2}
we conclude that in all cases the conditions of	
Lemma~\ref{lem: approx-ratio} hold,
so $x_1$ is a $\beta_1$-approximate solution for $w=w_1$.

Substitute $\beta_1 = 2\alpha(1+1/\gamma)^{k}$,
$\alpha = \sqrt{1+\epsilon/2}$ and
$\gamma = n^2 / \ln (\alpha)$,
and note $k\leq m \leq n^2$,
to get
\begin{align*}
	 \beta_1
	&\leq  2\alpha(1+1/\gamma)^{n^2}\\
	&=     2\alpha \left(1+ (\ln \alpha) /n^2 \right)^{n^2}\\
	& \leq  2\alpha e^{\ln \alpha}
	=     2+\epsilon.
\end{align*}

The desired approximation ratio is achieved.
\end{proof}

\subsection{Implementing \lrmwm{} in the Semi-Streaming Model}

In the previous section we showed that
\lrmwm{} computes a $(2+\epsilon)$-approximation for MWM.
In the semi-streaming model,
we must obey space constraints in addition to maintaining a good approximation ratio.
In the presentation of the sequential algorithm
we ignored the space constrains:
we did not remove edges from the stack,
and we represented the temporary solutions as the vectors $x_i$
of size $\Theta(n^2)$.

In order to follow the space constraints,
we replace any insertion of $\bot$ into the stack
by a removal of the relevant edge,
and the vectors $x_i$ by a single set containing the current matching.
For the sake of completeness,
we present \mwmsemi{} (Algorithm~\ref{alg: MWM semistream}),
an implementation of \lrmwm{} in the semi-streaming model.
The correctness of \mwmsemi{}
is derived directly from the correctness of \lrmwm{},
so we only need to prove it obeys the space constraints.

After omitting notations and auxiliary variables from \lrmwm{},
we are left only with three types of data structures
in \mwmsemi{}:
$M$ is the matching constructed,
$S$ is the stack
and $E(v)$ is a queue of edges from $S$ that
contain node $v$.
Every edge $(u,u')$ that is added to $S$
is also added to $E(u)$ and $E(u')$.
When $(u,u')$ is removed from $E(u)$ or from $E(u')$,
it is also removed from $S$,
implying $\size{S} \leq \sum_v \size{E(v)}$.
The next lemma bounds the size of $E(v)$ for every $v$.

\RestyleAlgo{boxruled}
\LinesNumbered
\begin{algorithm}[t]
	\DontPrintSemicolon
	\caption{\texttt{MWM-semi}($V,E,w$). A Semi-Streaming approximation algorithm for MWM}
	\label{alg: MWM semistream}
	$S \gets $ empty stack\\
	$\phi\gets \vec{0}$\\
	$\forall v\in V: E(v) \gets $ empty queue\\	
	\ForEach{$e=(u,u')\in E$%
		\label{lineb: foreach begin}}
	{
		
		\lIf{$w[e] \leq \alpha(\phi(u) + \phi(u'))$}
		{		
			continue
		}
		$S.\text{push}(e)$\\
		$w'[e]\gets w[e]-(\phi(u) + \phi(u'))$ \\
		\ForEach{$v \in e$}
		{
			$E(v).\text{enqueue}(e)$ \\
			$\phi(v) \gets \phi(v) + w'[e]$ \\
			
			\If{$ (\alpha-1)\alpha^{\size{E(v)}-2} >  2\alpha\gamma$}
			{\label{lineb: size ev cond}
				$e'\gets E(v).\text{dequeue}()$\\
				remove $e'$ from $S$
				\label{lineb: foreach end}\\
			}
		}		
	}	
	$M\gets \emptyset$ \\
	\While{$S \neq \emptyset$}
	{
		$e \gets S.\text{pop}()$\\
		\lIf{$M\cap N(e) = \emptyset$}
		{
			$M \gets M \cup \set{e}$
		}
	}
	return $M$	
\end{algorithm}

\begin{lemma}
\label{lem: alg-space}
During the execution of \mwmsemi{}, $|E(v)| =
O\left(\frac{\log n + \log (1/\epsilon)}{\epsilon}\right)$
for each $v\in V$.
\end{lemma}

\begin{proof}
After each iteration of the loop in
Lines~\ref{lineb: foreach begin}--\ref{lineb: foreach end},
we have $(\alpha-1)\alpha^{\size{E(v)}-2} \leq  2\alpha\gamma$ for each $v\in V$:
this is true at the beginning;
$E(v)$ can grow only by 1 at each iteration;
and whenever the inequality does not hold,
an edge is removed from $E(v)$.

From the above inequality, 
$\alpha = \sqrt{1+\epsilon/2}$, and $\gamma = n^2 / \ln (\alpha)$, 
we derive an asymptotic bound on $\size{E(v)}$.
\begin{align*}
\size{E(v)}
&\leq \frac{\log \frac{2\alpha \gamma}{\alpha-1}}{\log\alpha} +2\\ 
&\leq \frac{\log \frac{2n^2 \alpha^2}{(\alpha-1)^2}}{\log\alpha} +2
\tag{$\ln\alpha>\frac{\alpha-1}{\alpha}$ for $\alpha>1$} \\
&=    \frac{\log (2n^2) +2\log\alpha-2\log(\alpha-1)}{\log\alpha} +2\\
&\leq 6\frac{\log (2n^2) -2\log(\epsilon/6)}{\epsilon} +4
\tag{$\alpha-1>\epsilon/6$ and $\log\alpha>\epsilon/6$ for $0<\epsilon<6$} \\
&=    O\left(\frac{\log n + \log (1/\epsilon)}{\epsilon}\right)
\end{align*}
as desired.
\end{proof}

From Lemma~\ref{lem: alg-space}
we conclude that, for a constant $\epsilon$,
\mwmsemi{} maintains at most $O(n\log n)$ edges,
each represented by $O(\log n)$ bits,
giving a total space of $O(n\log^2 n)$ bits. 
Our algorithm requires $O(1)$ time to process a new edge arriving from the stream, 
and $O(n\log n)$ time for the post-processing step. 
A similar analysis, without assuming $\epsilon$ is constant, 
implies the main theorem of this paper.

\begin{theorem-repeat}{thm: the main theorem}
There exists an algorithm in the semi-streaming model computing a $(2+\epsilon)$-approximation for MWM, using $O(\epsilon^{-1}n\log n \cdot (\log n + \log (1/\epsilon)))$ bits and having an $O(1)$ processing time.
\end{theorem-repeat}

In our analysis we assume that the edge weights can be represented using $O(\log n)$ bits,
which is the case, e.g., for integer edge weights bounded by a polynomial in $n$. 
If this is not the case, and the weights are integers bounded by some $W$, our algorithm requires $O(n(\log^2 n + \log W))$ bits, as it keeps a sum of weights for every node,
and does not keep edge weights at all. 

\paragraph{Acknowledgments}
We thank Keren Censor-Hillel and Seri Khoury for helpful discussions,
Seffi Naor for useful comments on the presentation,
and the anonymous referees of SODA 2017 and ACM TALG for their comments.

\let\OLDthebibliography\thebibliography
\renewcommand\thebibliography[1]{
	\OLDthebibliography{#1}
	\setlength{\parskip}{0pt}
	\setlength{\itemsep}{2pt plus 0.3ex}
}
\bibliographystyle{alpha}
\bibliography{MM}

\end{document}